\documentclass[twocolumn,final,twoside,journal]{IEEEtran}

\usepackage{tabu}
\usepackage[usenames,dvipsnames]{xcolor}

\usepackage{amsmath,amsthm,graphicx,cite}
\usepackage{srcltx}
\usepackage{epsfig,amsfonts}
\usepackage{graphicx,cite,amssymb,amsmath}
\usepackage{color}

\usepackage{amsthm}
\usepackage[nolist]{acronym}
\usepackage{psfrag}
\usepackage{perso}
\usepackage{pgfplots}
 \pgfplotsset{compat=newest}
    \pgfplotsset{plot coordinates/math parser=false}
    \pgfplotsset{
    label style={anchor=near ticklabel},
    xlabel style={yshift=0.0em},
    ylabel style={yshift=-0.3em},
    tick label style={font=\footnotesize },
    label style={font=\footnotesize},
    legend style={font=\footnotesize},
    title style={font=\fontsize{7}}}
\usepackage{xcolor}
\definecolor{iso}{rgb}{0.7,0.7,0.7}
\usepackage{blindtext}
\usepackage{flushend}
\usepackage{relsize}
\usepackage{mathtools}
\mathtoolsset{showonlyrefs}

\usepackage{placeins}

\ifCLASSOPTIONcompsoc
\usepackage[caption=false,font=normalsize,labelfont=sf,textfont=sf]{subfig}
\else
\usepackage[caption=false,font=footnotesize]{subfig}
\fi

\usepgflibrary{arrows}
\usetikzlibrary{patterns}
\usepgflibrary{decorations.pathmorphing}

\newcommand{\rippleset}{\mathscr{R}}
\newcommand{\cloudset}{\mathscr{C}}

\newcommand{\Ripple}{\mathtt{R}}
\newcommand{\Cloud}{\mathtt{C}}
\renewcommand{\c}{\mathtt{c}}
\renewcommand{\r}{\mathtt{r}}

\newcommand{\Ripp}{\mathsf{R}}

\newcommand{\ripple}[1]{ \msr{R}_{#1}}
\newcommand{\cloud}[1]{\msr{C}_{#1}}

\newcommand{\ru}{\r_u}
\newcommand{\Ru}{\Ripple_u}
\newcommand{\Cu}{\Cloud_u}
\newcommand{\cu}{\c_u}

\renewcommand{\S}[1]{\mathtt{S}_{#1}}
\ifthenelse{\isundefined{\C}}{
    \newcommand{\C}[1]{\mathtt{C}_{#1}}
    }{
    \renewcommand{\C}[1]{\mathtt{C}_{#1}}
    }

\ifthenelse{\isundefined{\N}}{
    \newcommand{\N}[1]{\mathsf{N}_{#1}}
    }{
    \renewcommand{\N}[1]{\mathsf{N}_{#1}}
    }

\newcommand{\Erv}{\mathtt{A}}
\newcommand{\erv}{\mathtt{a}}
\renewcommand{\b}{\mathtt{b}}
\newcommand{\B}{\mathtt{B}}
\newcommand{\n}{\mathtt{n}}

\renewcommand{\nu}{\n_u}

\newcommand{\bu}{\b_u}
\newcommand{\Bu}{\B_u}

\newtheorem{mydef}{Definition}

\newtheorem{theorem}{Theorem}

\newcommand{\myop}[1]{%
  \mathchoice{\raisebox{8pt}{$\displaystyle #1$}}
             {\raisebox{8pt}{$#1$}}
             {\raisebox{4pt}{$\scriptstyle #1$}}
             {\raisebox{1.6pt}{$\scriptscriptstyle #1$}}}

\newcommand{\paccess}{p}
\newcommand{\pu}{q_u}

\newcommand{\slot}{y} 

\newcommand{\per}{{\mathsf{P}}_{\mathsf{e}}}
\newcommand{\throughput}{\mathsf{T}}
\newcommand{\nuser}{n}
\newcommand{\nslot}{m}

\newcommand{\betamax}{\beta_{\max}}

\begin{document}
\begin{acronym}
\acro{RA}{Random Access}
\acro{LT}{Luby Transform}
\acro{BP}{belief propagation}
\acro{i.i.d.}{independent and identically distributed}
\acro{PER}{Packet Error Rate}
\end{acronym}

\title{Finite-Length Analysis of Frameless ALOHA}

\author{
    \IEEEauthorblockN{Francisco L\'azaro\IEEEauthorrefmark{1}, \v Cedomir Stefanovi\' c\IEEEauthorrefmark{2}\\
    \IEEEauthorblockA{\IEEEauthorrefmark{1}Institute of Communications and Navigation of DLR (German Aerospace Center),
    \\Wessling, Germany. Email: Francisco.LazaroBlasco@dlr.de}\\
    \IEEEauthorblockA{\IEEEauthorrefmark{2}Department of Electronic Systems, Aalborg University
    \\Aalborg, Denmark. Email: cs@es.aau.dk}\\
    \thanks{This work has been accepted for publication at the 11th International {ITG} Conference on Systems, Communications and Coding, SCC 2017.}
\thanks{\copyright 2016 IEEE. Personal use of this material is permitted. Permission
from IEEE must be obtained for all other uses, in any current or future media, including
reprinting /republishing this material for advertising or promotional purposes, creating new
collective works, for resale or redistribution to servers or lists, or reuse of any copyrighted
component of this work in other works}
}
}
\maketitle



\thispagestyle{empty} \pagestyle{empty}

\begin{abstract}
In this paper we present an exact finite-length analysis of frameless ALOHA that is obtained through a dynamical programming approach. Monte Carlo simulations are performed in order to verify the analysis. Two examples are provided that illustrate how the analysis can be used to optimize the parameters of frameless ALOHA.
To the best of the knowledge of the authors, this is the first contribution dealing with an exact finite-length characterization of a protocol from the coded slotted ALOHA family of protocols.
\end{abstract}





\section{Introduction}\label{sec:Intro}

Typical networking scenarios commonly involve uncertainty in terms of the time-instants when the activation of the communicating devices will happen, i.e., when the need for  communication will arise.
In scenarios in which multiple devices share a common medium and an access point (AP), this type of uncertainty is typically resolved via means of a random access protocol.
Slotted ALOHA \cite{R1975} is an example of such protocol, in which devices, upon activation, randomly and independently choose time slots in which they attempt connecting (i.e., contend for the access) to the AP by transmitting their packets.
A collision of two or more packets is typically considered destructive, i.e., all packets involved in a collision are lost. Hence, only slots that contain a single packet (singleton slots) are considered useful and the related packets successfully received.
In the above model, known as collision channel model, the asymptotic maximal throughput of slotted ALOHA, defined as the probability of successfully receiving a packet in a slot, is only $1/e$, implying that most of the slots are wasted.

The introduction of successive interference cancellation in the above framework significantly changed the perspective on the capabilities of random access protocols \cite{CGH2007}.
Namely, assume that users are sending multiple replicas of the same packet when contending, embedding in each replica information that allows to determine the position of all other replicas of the same packet.
A packet that occurs in a singleton slot is successfully received, enabling the AP to identify the slots in which the other replicas occurred and to remove the replicas from the stored waveform using an interference cancellation algorithm, see Fig.~\ref{fig_1}.
Subsequently, some of the collision slots become singletons, promoting the recovery of new packets and the removal of their replicas.
This process is analogous to the iterative belief-propagation decoding of erasure-correcting codes, enabling the use of theory and tools of codes-on-graphs to design and analyze slotted ALOHA-based protocols \cite{L2011}.
In fact, for the collision channel model the asymptotic throughput can be pushed to the ultimate limit of 1 packet per slot \cite{PLC2015}.
The price to pay is that the AP has to buffer the received signal, and it has to employ more complex signal processing, required for the interference cancellation.

\begin{figure}[t]\centering
	\includegraphics[width=0.7\columnwidth]{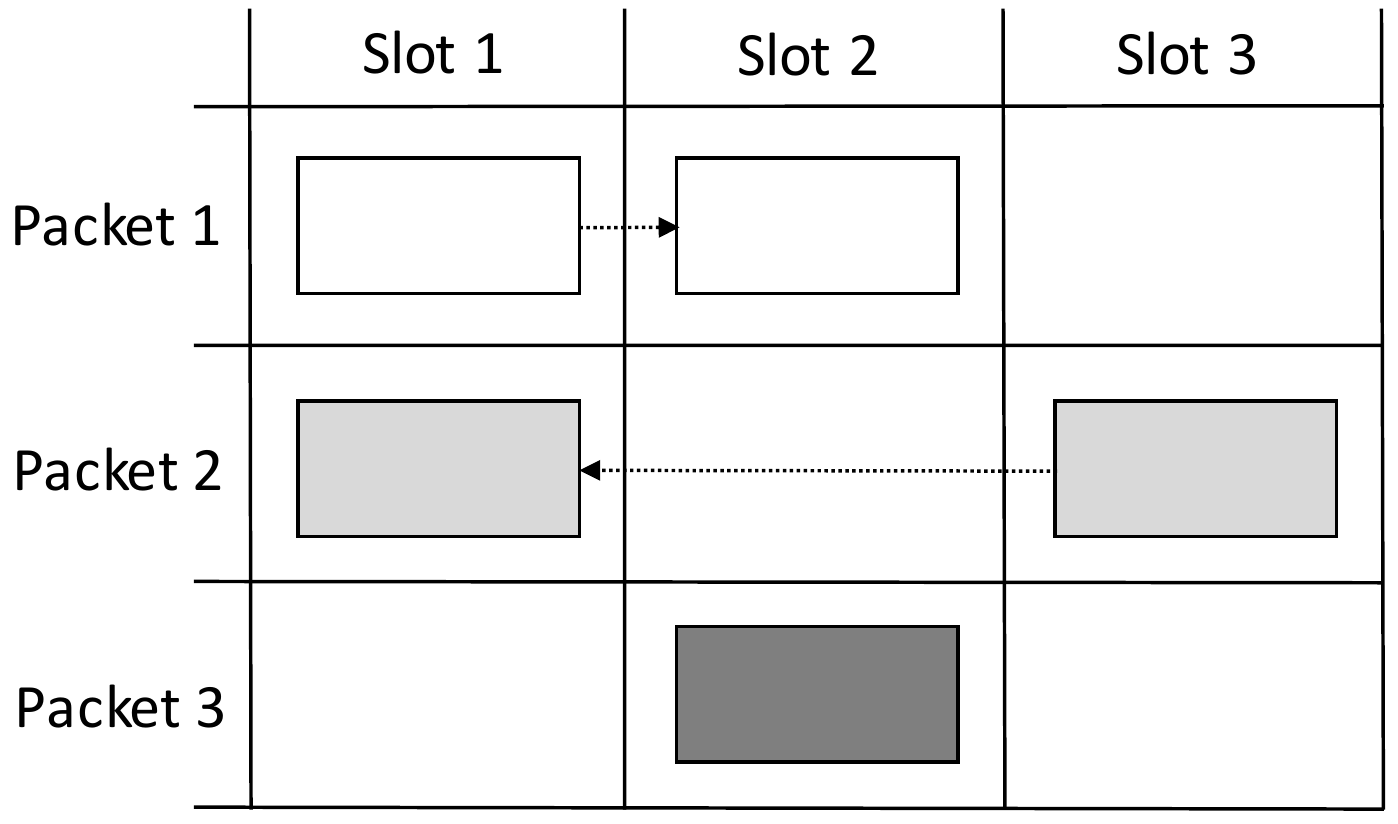}
	\caption{Example of SIC-enabled slotted ALOHA. Packet 2 is received in singleton slot 3, enabling its recovery and the cancellation of its replica from slot 1. In turn, slot 1 becomes singleton and packet 1 is recovered from it, enabling its cancellation from slot 2. Slot 2 becomes singleton and packet 3 is recovered from it.}
	\label{fig_1}
\end{figure}

The results presented in \cite{L2011} inspired a strand of works that applied various concepts from codes-on-graphs to design SIC-enabled slotted ALOHA schemes\cite{PLC2011,LPLC2012,SPV2012,PSLP2014}, which are usually referred to using the umbrella term of coded slotted ALOHA.
In this paper we focus on frameless ALOHA \cite{SPV2012,SP2013}, which exploits ideas originating from the rateless coding framework \cite{luby02:LT}.
In particular, frameless ALOHA is characterized by (i) a contention period that consists of a number of slots that is not defined a priori, but terminated when the number of resolved\footnote{Under user resolution we assume recovery/decoding of user packet.} users and/or instantaneous throughput reach certain thresholds and (ii)
a slot access probability with which a user decides on a slot basis whether to transmit or not its packet replica.
An asymptotic optimization of the slot access probability that maximizes the expected throughput was performed in \cite{SPV2012}.
A joint assessment of the optimal slot access probability and the contention termination criteria in non-asymptotic, i.e., finite-length scenarios were assessed by means of simulations in \cite{SP2013}.
Thus, so far the finite-length performance of frameless ALOHA, like for other SIC-enabled slotted ALOHA protocols, could only be established by means simulations or approximate methods that are usually only accurate in the error floor region, see \cite{ivanov:floor}, for example.

In this paper we build up on the approach for finite-length analysis of rateless codes presented in \cite{Karp2004,lazaro:Allerton2015}, applying it to the context of frameless ALOHA.
Specifically, we present an exact finite-length analysis of frameless ALOHA for the collision channel with successive interference cancellation at the receiver.
This analysis provides the average throughput and the packet error rate, and, to the best of the knowledge of the authors, is the first exact finite-length analysis of a random access protocol belonging to the coded slotted ALOHA family of protocols.
Furthermore, we illustrate by means of examples how the proposed analysis can be used to optimize the parameters of  frameless ALOHA.

The rest of the paper is organized as follows. In Section~\ref{sec:sysmodel} we describe the system model assumed in this paper. In Section~\ref{sec:finiteLA} we present a finite length analysis of frameless ALOHA. Section~\ref{sec:opt} presents two examples in which the presented analysis is used to optimize the parameters of frameless ALOHA in the finite-length regime. Finally, Section~\ref{sec:Conclusions} presents the conclusions.


\section{System Model}\label{sec:sysmodel}

We adopt a simple model of traffic arrivals, focusing on the analysis of the performance of the proposed contention mechanism from the system point of view.
Specifically, we consider a single instance of batch arrival of $\nuser$ users, which contend for the access to a single access point, Fig.~\ref{fig_2}(a).

Contention is performed during a period that consists of $\nslot$ slots, where $\nslot$ is not defined a-priori but determined on the fly.
The users are slot and contention period synchronous, and they arrive prior to the start of the contention period.
A user contends by transmitting replicas of the same packet; for each slot of the contention period a user decides with slot access probability $\paccess$ whether to transmit a replica or not, independently of any other slot and of any other user, see Fig~\ref{fig_2}(b).
%


\begin{figure}[t]
        \centering
        \subfloat[System model] {\includegraphics[width=0.4\columnwidth]{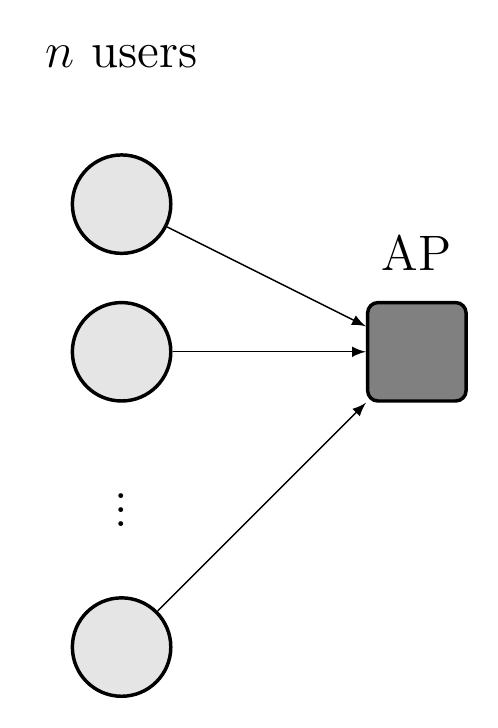}
        \label{fig_2a} }
        \vline
        \subfloat[Contention model] { \includegraphics[width=0.4\columnwidth]{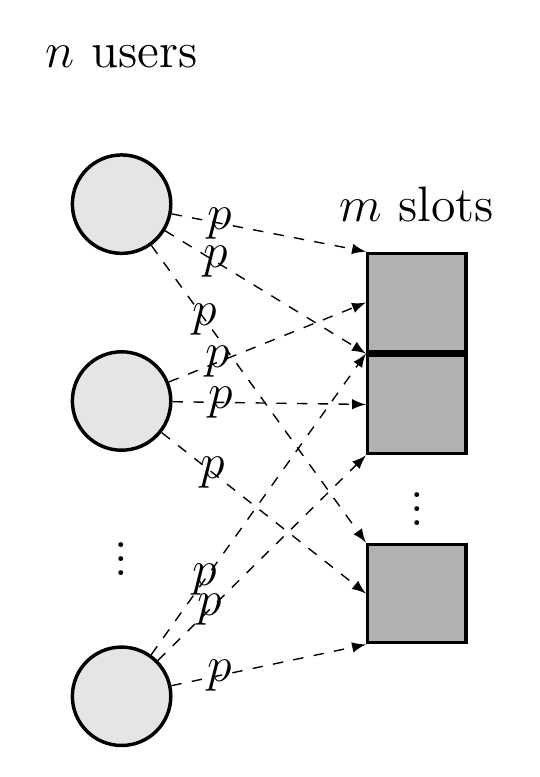}
        \label{fig_2b} }
        \caption{System and contention models assumed in this work.}
        \label{fig_2}
\end{figure}

For the sake of simplicity, we initially assume that $\paccess$ is uniform over users and slots and that
\begin{align}
\paccess = \frac{\beta}{\nuser}
\end{align}
where $\beta$ is a suitably chosen constant.
Thus, the contention on a user basis is modelled as a Bernoulli trial with $\nslot$ repetitions with probability $\paccess$. Hence, the number of replicas that a user sends follows also a binomial distribution with parameters $\nslot$ and $\paccess$.
Denoting as $\Omega= (\Omega_0,\Omega_1,~\Omega_2,~\Omega_3,~\hdots~\Omega_{\nuser})$ the slot degree distribution, where $\Omega_i$ corresponds to the probability of a slot having degree $i$, we have
\[
\Omega_i = \binom{\nuser}{i} \paccess^i (1-\paccess)^{\nuser-i}.
\]
In practice, $\Omega$ can be tightly approximated by a Poisson distribution whenever $\nuser$ takes moderate or large values, i.e.,
\begin{align}
\Omega_i \approx  \frac{( n p )^i}{i!} e^{-np} = \frac{\beta^i}{i !} e^{-\beta}.
\end{align}

We assume a widely adopted collision channel model, so that slots containing only one transmission (singleton slots) can be decoded with probability $1$ and slots containing more
than one transmission are undecodable with probability $1$. Following \cite{L2011} we describe the iterative successive cancellation process at the receiver using a bipartite graph, assuming $\nslot$ is fixed.
By $\mathbf{v}=(v_1,~v_2,\ldots, v_\nuser)$ we denote the $\nuser$ users, and by $\mathbf{\slot}=(\slot_1, \slot_2, \ldots, \slot_\nslot)$ the $\nslot$ slots.
We use the notation $\deg(\slot)$ to refer to the (original) degree of a slot, that is, the number of users that transmitted in the slot.
Furthermore, during the decoding (i.e., SIC) process we introduce the term reduced degree to refer to the number of unresolved user packets that are still present in the slot.
Thus, the reduced degree of a slot will be equal of less than its (original) degree.
The following definitions are used in the finite-length analysis.
\begin{mydef}[Ripple] We define the ripple as the set of singleton slots (reduced degree 1) and we denote it by $\rippleset$.
\end{mydef}
\noindent The cardinality of the ripple is denoted by $\r$ and its associated random variable as $\Ripp$.
\begin{mydef}[Cloud] We define the cloud as the set of slots with reduced degree $d\geq 2$ and we denote it by $\cloudset$.
\end{mydef}
\noindent The cardinality of the cloud is denoted by $\c$ and the corresponding random variable as $\Cloud$.

For illustration, in Fig.~\ref{fig:graph} we provide an example of bipartite graph for $\nuser=4$ users and $\nslot=4$ slots. We can observe how slots $\slot_1$ and $\slot_4$ belong to the ripple and slots $\slot_2$ and $\slot_3$ belong to the cloud.

\begin{figure}[t]\centering
	\includegraphics[width=0.99\columnwidth]{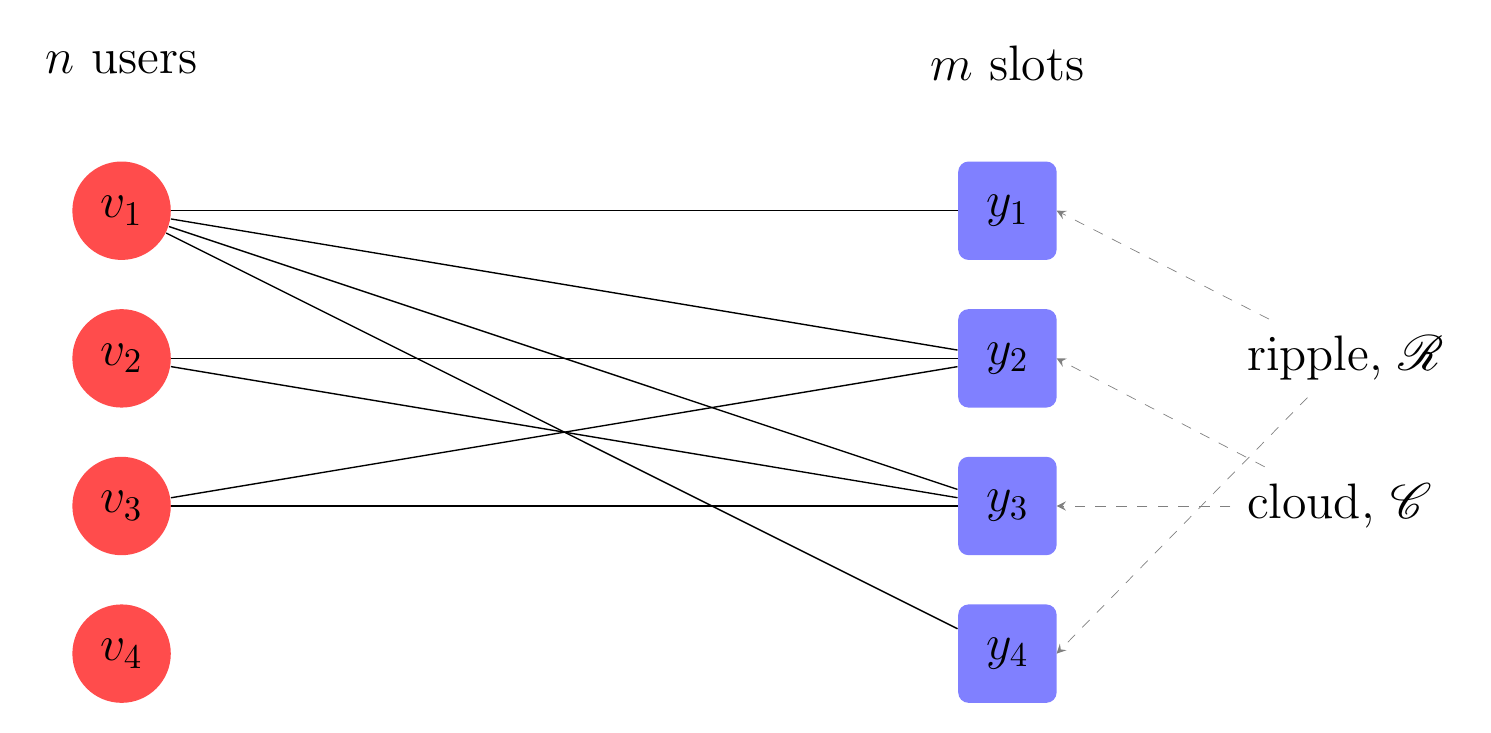}
	\caption{Bipartite graph representation of slotted ALOHA.}
	\label{fig:graph}
\end{figure}

In the next sections, in the ripple and cloud we will add a temporal dimension through the subscript $u$ that corresponds to the number of unresolved users.
Initially, all $\nuser$ users are unresolved, hence $u=\nuser$.
At each step, if the ripple is not empty exactly 1 user gets resolved, and thus the subscript decreases by 1.
After $\nuser$ decoding steps, all users are resolved and decoding ends, i.e., $u=0$.
If at any of the $\nuser$ decoding steps the ripple is empty, decoding fails.


\section{Finite-Length Analysis}
\label{sec:finiteLA}

In this section we follow the approach in \cite{Karp2004,lazaro:Allerton2015} introduced for LT codes and model the iterative successive cancellation process of frameless ALOHA by means of a finite state machine with state
\[
\S{u}:=(\Cu, \Ru )
\]
i.e., the state comprises the cardinalities of the cloud and the ripple at the decoding step  in which $u$ users are unresolved.
The following theorem establishes a recursion that can be used to determine the decoder state distribution.

\begin{theorem}\label{theorem:rec}
Given that the decoder is at state $\S{u}=(\cu,\ru)$, when $u$ users are unresolved and with $\ru>0$, the probability of the decoder being at state $\Pr \{ \S{u-1}=(\c_{u-1}, \r_{u-1}) \}$ when $u-1$ users are unresolved is given by
\begin{align}
\Pr\{\S{u-1}&=(\cu-\bu,\ru-\erv_u+\bu) | \S{u}=(\cu,\ru)\} = \nonumber \\[2mm]
& \binom{\cu}{\bu} {\pu}^{\bu} (1-\pu)^{\cu-\bu} \binom{\ru-1}{\erv_u-1} \,\,\times \nonumber \\[2mm]
&\left(\frac{1}{u}\right)^{\erv_u-1} \left( 1- \frac{1}{u} \right)^{\ru-\erv_u}
\label{eq:prob_transition}
\end{align}
for $0 \leq \bu \leq \cu$, $\erv_u - \bu \leq \ru$ and $\erv_u \geq 1$, and with
\begin{align}
\pu =  \frac{ \mathlarger {\sum}\limits_{d=2}^{\nuser-u+2}   \Omega_d \, d (d-1)\frac{1}{\nuser}  \frac{u-1}{\nuser-1}  \frac{\binom{\nuser-u}{d-2}}{\binom{\nuser-2}{d-2}} }
{  1 -  \mathlarger{\sum}\limits_{d=1}^{\myop{\nuser-u+1}}    \Omega_d \,u \frac{\binom{\nuser-u}{d-1}}{\binom{\nuser}{d}}
-  \mathlarger{\sum}\limits_{d=0}^{\myop{\nuser-u}}   \Omega_d \frac{\binom{\nuser-u}{d}}{\binom{\nuser}{d}}    }.
\label{eq:pu_theorem}
\end{align}

\end{theorem}
\begin{proof}
The proof reduces to analyzing the variation of the cloud and ripple sizes in the transition from $u$ to $u-1$ unresolved users. Since we assume  $\ru>0$, in the transition from $u$ to $u-1$ unresolved users exactly 1 user is resolved. All the edges coming out from the resolved user are erased from the decoding graph. As a consequence some slots might leave the cloud and enter the ripple if their reduced degree becomes one, and other slots will leave the ripple if their reduced degree decreases from 1 to 0.

Let us first focus of the number of slots leaving $\cloud{u}$ and entering $\ripple{u-1}$ in the transition, denoted by $\bu$ and with associated random variable given by $\Bu$. Due to the nature of frameless ALOHA, in the decoding graph every slot chooses its neighbor users uniformly at random and without replacement. Thus, random variable $\Bu$ is binomially distributed with parameters $\c_u$ and $\pu$, being $\pu$ the probability of a generic slot $\slot$ leaving $\cloud{u}$ to enter $\ripple{u-1}$,
\begin{equation}
\pu := \Pr \{ \slot \in \ripple{u-1} | \slot \in \cloud{u} \}= \frac { \Pr \{ \slot \in \ripple{u-1}\, , \, \slot \in \cloud{u} \} }  { \Pr \{ \slot \in \cloud{u} \}}.
\label{eq:pu_prob}
\end{equation}
We shall first focus on the enumerator of \eqref{eq:pu_prob} and we shall condition it to the slot having degree $d$,
${\Pr \{ \slot \in \ripple{u-1}\, , \, \slot \in \cloud{u} | \deg(\slot)= d \}}$. This corresponds to the probability that one of the $d$ edges of  slot $\slot$ is connected to the user being resolved at the transition, one edge is connected to one of the $u-1$ unresolved users after the transition and the remaining $d-2$ edges are connected to the $\nuser-u$ unresolved users before the transition. In other words, slot $\slot$ must have \emph{reduced} degree $2$ \emph{before} the transition and \emph{reduced} degree $1$ \emph{after} the transition.
It is easy to see how this probability corresponds to
\begin{align}
 \Pr \{ \slot \in \ripple{u-1}\, , \, & \slot \in \cloud{u} | \deg(\slot)= d \} = \nonumber \\
& \mkern-65mu \frac{d}{\nuser} (d-1)\frac{u-1}{\nuser-1}  \frac{\binom{\nuser-u}{d-2}}{\binom{\nuser-2}{d-2}}
\label{eq:z_and_l_d}
\end{align}
for $d\geq 2$. In the complementary case, $d < 2$, it is obvious that the slot cannot enter the ripple. Thus, we have
\[
\Pr \{ \slot \in \ripple{u-1}\, , \, \slot \in \cloud{u} | \deg(\slot)= d \} = 0
\]
for $d<2$. 

Let us now concentrate on the denominator of \eqref{eq:pu_prob}, that corresponds to the probability that a slot $\slot$ is in the cloud when $u$ users are still unresolved. This is equivalent to the probability of slot $\slot$ not being in the ripple or having reduced degree zero (all edges connected to resolved users). Hence, we have
\begin{align}
\Pr  \{ \slot \in \cloud{u}\}&=   1 -  \mathlarger{\sum}_{d=1}^{\nuser}  \Omega_d   u\frac{\binom{\nuser-u}{d-1}}{\binom{\nuser}{d}} - \mathlarger{\sum}_{d=0}^{\nuser}  \Omega_d \frac{\binom{\nuser-u}{d}}{\binom{\nuser}{d}}
\label{eq:z}
\end{align}
where the first summation on the right hand side corresponds to the probability of a slot being the ripple, and the second summation corresponds to the probability of a slot having reduced degree zero.
Inserting \eqref{eq:z_and_l_d} and \eqref{eq:z} in \eqref{eq:pu_prob}, the expression of $\pu$ in \eqref{eq:pu_theorem} is obtained, and the variation of size of the cloud, i.e., random variable $\Bu$, is determined.

We focus next on the variation of size of the ripple in the transition from $u$ to $u-1$ resolved users. In the transition some slots enter the ripple ($\bu$ slots) but there are also slots leaving the ripple. Let us denote by $\erv_u$ the number of slots leaving the ripple in the transition from $u$ to $u-1$ unresolved users, and let us refer to the associated random variable as $\Erv_u$. Assuming that the ripple is not empty\footnote{If the ripple is empty, $\ru=0$, no slots can  leave the ripple. Moreover, decoding stops, so there is no transition.}, the decoder will select at uniformly random one slot from the ripple, that we denote as $\slot$. The only neighbour of $\slot$, $c$ will get resolved. All slots in the ripple that are connected to $c$ leave the ripple in the transition. Hence, we have that slot $\slot$ leaves the ripple. Additionally, the remaining $\ru-1$ slots in the ripple will leave the ripple with probability $1/u$, which is the probability that they have $c$ as neighbour. Thus, the probability mass function of $\Erv_u$ is given by
\begin{align}
\Pr\{\Erv_u=\erv_u & |\Ru=\ru\}=\\ &\binom{\ru-1}{\erv_u-1} \left(\frac{1}{u}\right)^{\erv_u-1} \mkern-3mu \left( 1- \frac{1}{u} \right)^{\ru-\erv_u}.
\end{align}
The proof is completed simply by observing that by definition,
\[
{\r_{u-1}= \ru -\erv_u + \bu}
\]
and
\[
\c_{u-1} = \cu - \bu.
\]
\end{proof}

The initial state of the decoder corresponds to a multinomial distribution with $\nslot$ experiments (slots) and three possible outcomes for each experiment, the slot being in the cloud, the ripple or having degree 0, with respective probabilities $(1-\Omega_1 - \Omega_0)$, $\Omega_1$ and $\Omega_0$. Thus, we have
\begin{align}
 \Pr\{\S{\nuser} =(\c_{\nuser},\r_{\nuser}) \} =\\ & \mkern-135mu \frac{\nslot!}{\c_\nuser! \, \r_\nuser!\, (\nslot-\c_\nuser-\r_\nuser)!} \left( 1-\Omega_1- \Omega_0\right)^{\c_{\nuser}}\, \Omega_1^{\r_{\nuser}} \, \Omega_0^{\nslot-\c_\nuser-\r_\nuser}
 \label{eq:init_state}
\end{align}
for all non-negative $\c_{\nuser},\r_{\nuser}$ such that $\c_{\nuser}+\r_{\nuser} \leq \nslot$.

The decoder state probabilities can be determined by initializing the finite state machine according to \eqref{eq:init_state} and applying recursively Theorem~\ref{theorem:rec}.
However, rather than in state probabilities, in random access one is interested in the \ac{PER}, i.e., the probability that a user is not resolved when the decoding process ends, denoted as $\per$. By observing that the decoding process ends at stage $u$ whenever $\ru=0$, leaving $u$ users unresolved and $\nuser-u$ resolved users, we have
\begin{align}
\per = \sum_{u=1}^{\nuser} \sum_{\cu} \frac{u}{\nuser} \Pr\{\S{u} =(\cu,0) \}.
\end{align}
A complementary figure of metric is the throughput, denoted by $\throughput$, which is the number of resolved users normalized by the number of slots. The expression of $\throughput$ corresponds to
\begin{align}
\throughput = \frac{\nuser (1-\per)}{m}=\frac{1-\per}{\nslot/\nuser}.
\label{eq:T}
\end{align}

Fig.~\ref{fig:throughput_test} and Fig.~\ref{fig:per_test} show, respectively, the throughput $\throughput$, and packet error rate $\per$ as a function of $\nslot/\nuser$ for $\beta = 2.5$ and $\nuser=100$. The figures show the outcome of the finite state machine analysis as well as the outcome of simulations. The simulation results were obtained by simulating 10000 contention periods. It can be observed how the match is exact down to the simulation error.

\begin{figure}
    \centering
	\includegraphics[height=7.7cm]{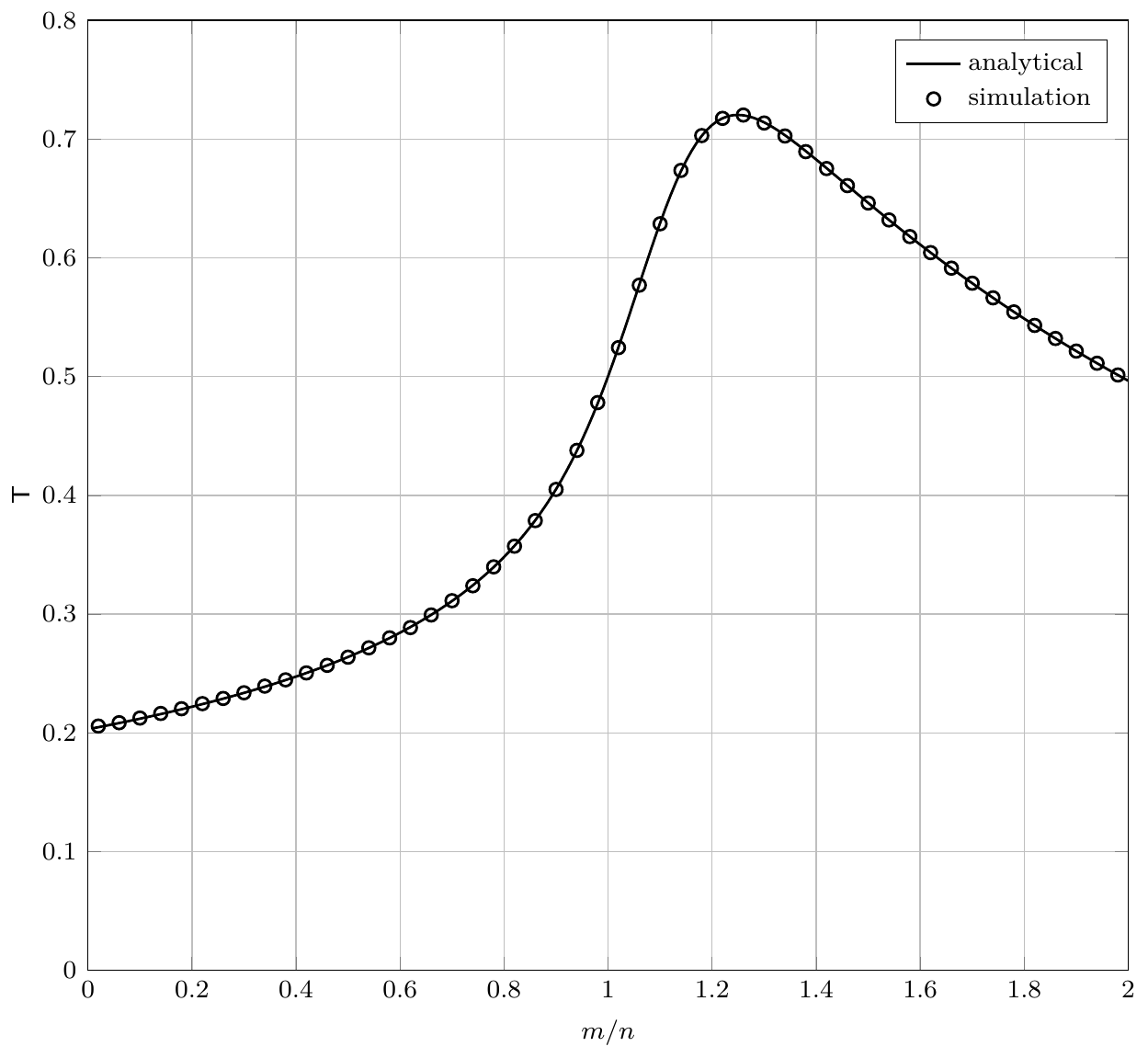}
	\caption{Throughput $\throughput$ as a function of $\nslot/\nuser$ for $\beta = 2.5$, for $\nuser=100$. The solid line represents was obtained using the finite state machine analysis and the markers are the outcome of Monte Carlo simulations.}
	\label{fig:throughput_test}
\end{figure}

\begin{figure}
\centering
	\includegraphics[height=7.7cm]{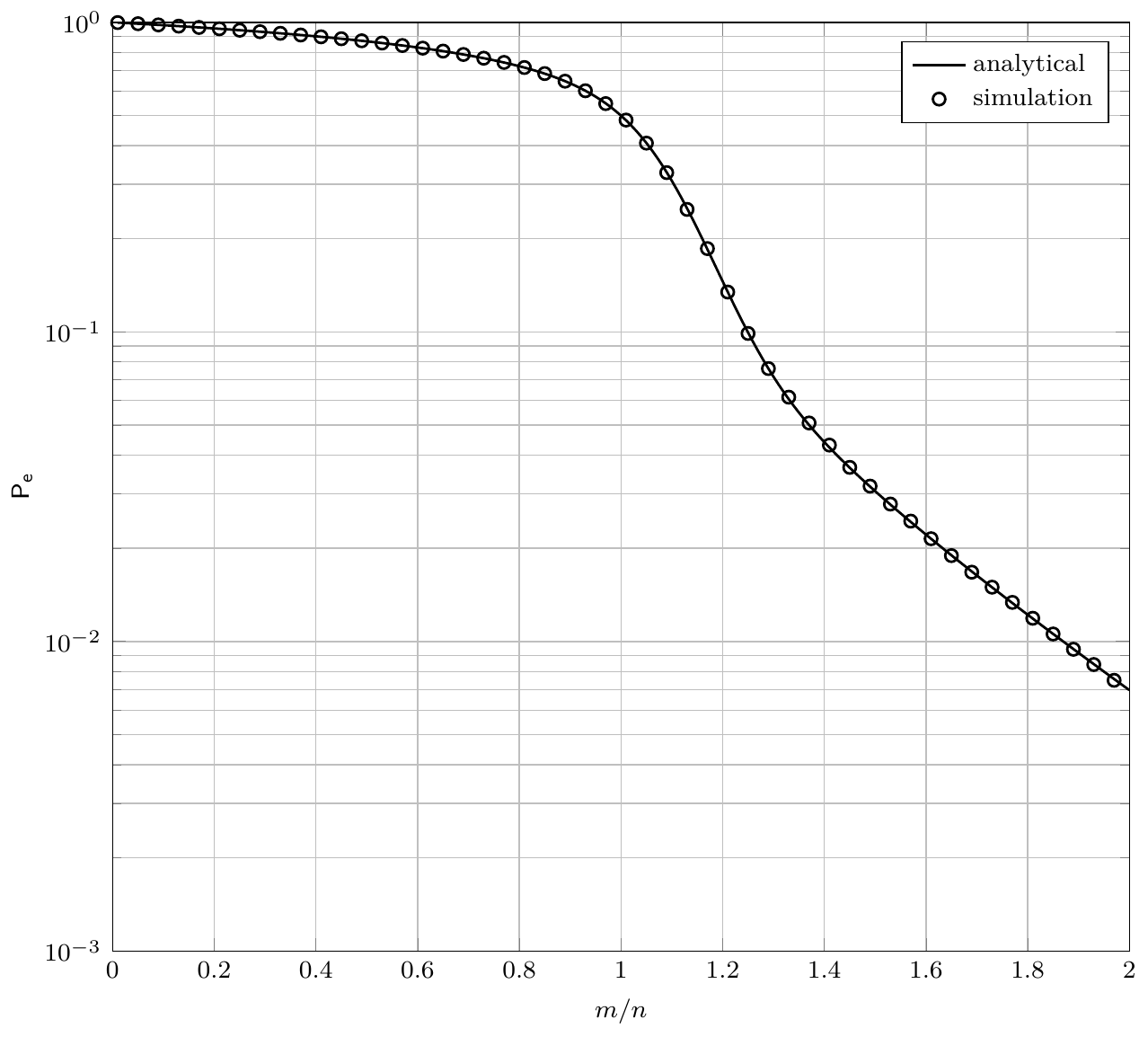}
	\caption{Packet error rate $\per$ as a function of $\nslot/\nuser$ for $\beta = 2.5$, for $\nuser=100$. The solid line represents was obtained using the finite state machine analysis and the markers are the outcome of Monte Carlo simulations.}
	\label{fig:per_test}
\end{figure}


\section{Optimization}\label{sec:opt}
As shown in \cite{SPV2012}, the throughput $\throughput$ of frameless ALOHA typically shows a behavior like the one in Fig.~\ref{fig:throughput_test}, where three distinct phases can be recognised.
Initially, $\throughput$ increases slowly with $\nslot/\nuser$.
When $\nslot/\nuser$ becomes close to $1$, the successive interference cancellation process kicks in and $\throughput$ increases sharply till a maximum value is achieved.
In the third phase, the throughput $\throughput$ starts decreasing.
This decrease is caused by the fact that $\per $ gets close to $0$, see Fig.~\ref{fig:per_test}.
Thus, when $\nslot/\nuser$ increases the throughput $\throughput$ decreases approximately as $1/(\nslot/\nuser)$.
This last phase coincides with the error floor of the packet error rate $\per$, as depicted in Fig.~\ref{fig:per_test}.

In some applications it is desirable to maximize the (overall) throughput $\throughput$ in order to use the channel as efficiently as possible.
However, in other settings it might be more important to minimize the packet error rate $\per$ for the given ratio of $\nslot / \nuser$.
In the following we provide two examples that illustrate how the analysis presented in Section~\ref{sec:finiteLA} can be used to customise the frameless ALOHA protocol. 

\subsection{Peak Throughput Maximization}\label{sec:opt_max_peak}

In this section we use the finite-length analysis to find the value of $\beta$ that maximizes the peak throughput $\throughput$, which we refer to as $\beta_{\max}$.
Formally, given $\nuser$ and treating the throughput as a bivariate function of $\beta$ and $\nslot/\nuser$, i.e., $\throughput(\beta,\nslot/\nuser)$, we define $\beta_ {\max}$ as:
\[
\betamax = \argmax_{\beta}~ \throughput_{\max} (\beta)
\]
where
\[
\throughput_{\max} (\beta)= \max_{\nslot/\nuser}  \throughput(\beta,\nslot/\nuser).
\]

The optimization was carried out for $\nuser=50,100$ and $200$. The result of the optimization can be found in Table~\ref{tab:params_peak}, where
$\nslot_{\max}$ is the number of slots that maximizes the throughput for $\beta= \betamax$.
These results are inline with the results obtained by means of simulation in \cite{SPV2012}; in particular, the trend is that  $ \betamax$ increases with $\nuser$, where the optimal $\betamax \approx 3.1$ when $n \rightarrow \infty$ \cite{SPV2012}.

\begin{figure}
\centering
	\includegraphics[height=7.6cm]{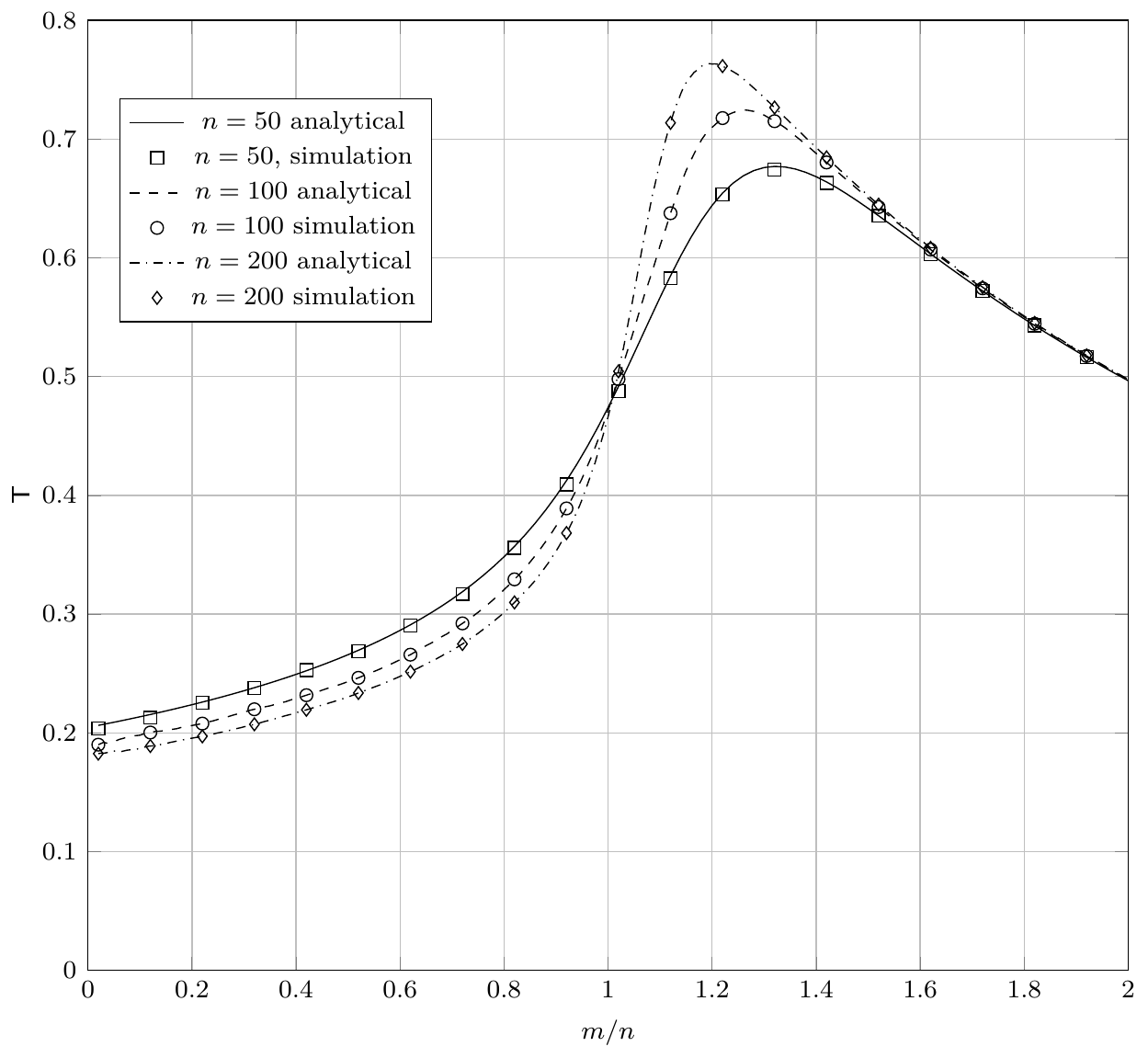}
	\caption{Throughput $\throughput$ as a function of $\nslot/\nuser$ for $\nuser=50, 100$ and $200$ and $\beta=\betamax$. The lines represent the outcome of the finite-length analysis and the markers the result of Monte Carlo simulations. }
	\label{fig:throughput_opt}
\end{figure}

\begin{table}[t]
\centering
\caption{Optimal parameters for frameless ALOHA}
    \begin{tabular}{|c|c|c|c|}
    \hline
    $\nuser$        & 50 & 100 & 200 \\ \hline \hline
    $\betamax$     & 2.47 & 2.62 & 2.71 \\ \hline
    $\throughput_{\max} (\betamax)$ & 0.67 & 0.72 & 0.76 \\ \hline
    $\nslot_{\max}$  & 66 & 126 & 240 \\ \hline
    \end{tabular}
\label{tab:params_peak}
\end{table}

\begin{figure}
\centering
	\includegraphics[height=7.6cm]{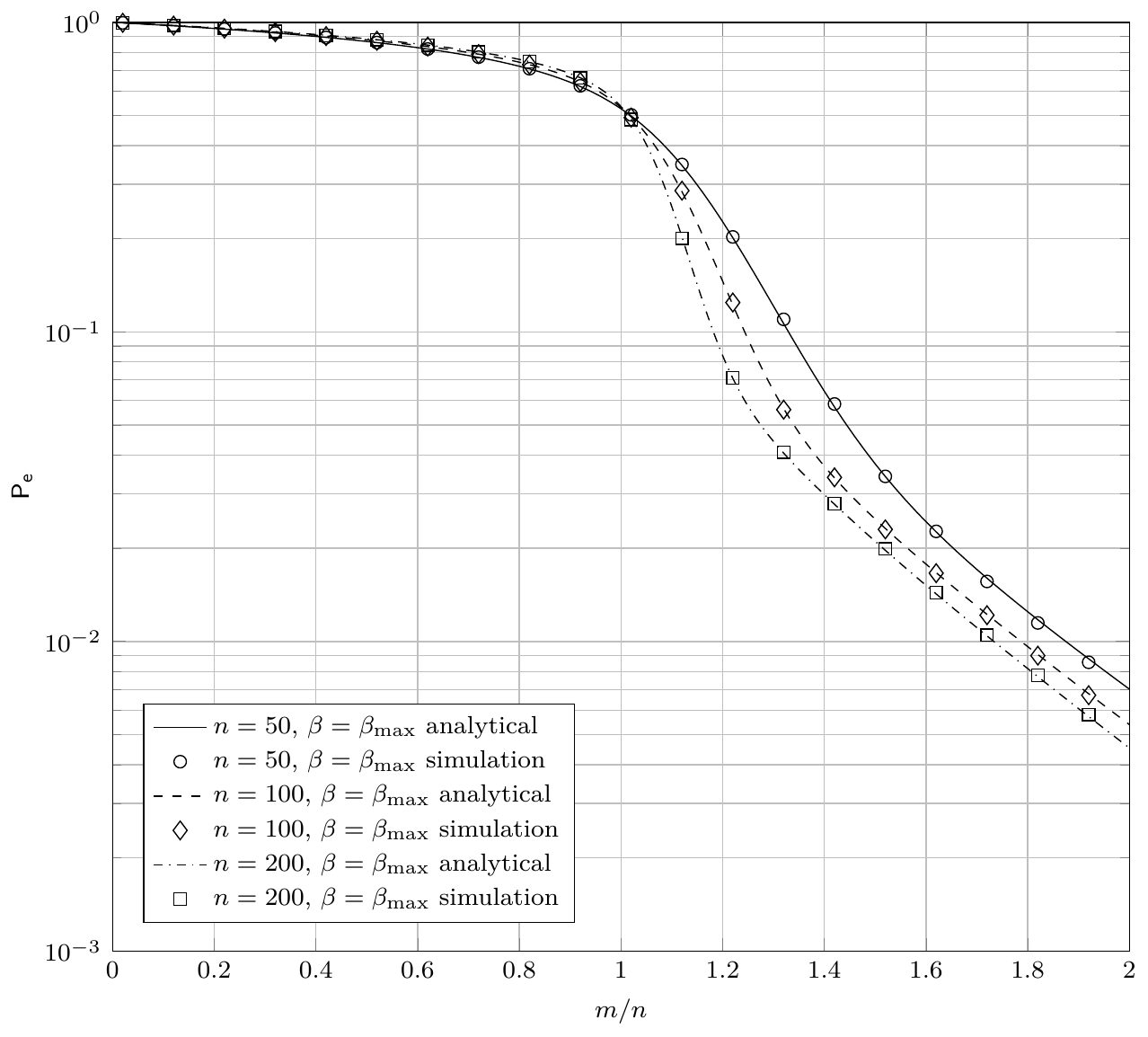}
	\caption{Throughput $\throughput$as a function of $\nslot/\nuser$ for $\nuser=50, 100$ and $200$ and $\beta=\betamax$. The lines represent the outcome of the finite-length analysis and the markers the result of Monte Carlo simulations. }
	\label{fig:per_opt_peak}
\end{figure}

For illustration, we also present the evolution of the throughput $\throughput$ obtained for $\beta=\betamax$, and $\nuser=50,100$ and $200$ in Fig.~\ref{fig:throughput_opt}, where the corresponding loci of the maxima are given in Table~\ref{tab:params_peak}. The figure shows the outcome of the finite-length analysis and also the result of Monte Carlo simulations. Concretely, for each value of $\nuser$ and $\nslot$ 10000 contention periods were simulated.
The related evolution of the PER is shown in Fig.~\ref{fig:per_opt_peak}, 
we can observe how when $\betamax$ is used $\per$ presents a rather high error floor, implying that maximization of the peak throughput does not guarantee a favorable PER performance.
Specifically, a lower bound on PER is given by
\begin{align}
\per \geq { \nslot \choose \nuser } \left( 1 - \frac{\beta}{n} \right)^{\nslot} 
\label{eq:lower_bound}
\end{align}
where the expression on the right-hand side refers to the probability that a user does not transmit at all in any of the slots. For large $\nslot$ one can approximate the lower bound in \eqref{eq:lower_bound} as
\begin{align}
\per \gtrsim e^{ - \beta \frac{\nslot}{\nuser}}
\label{eq:lower_bound1}
\end{align}
Thus, if we fix $\nslot/\nuser$, the error floor can be lowered by increasing $\beta$.
In the following section, we show how this error floor can be lowered substantially by increasing $\beta$ after the peak throughput has been reached.

\subsection{Pushing Down the Error Floor}\label{sec:opt_min_per}

In this section we consider frameless ALOHA schemes where the parameter $\beta$ varies with the slot number.
Concretely, we consider
\begin{equation}
\beta =
\begin{cases}
\beta_1, \qquad  \text{if } \nslot \leq \nslot^* \\
\beta_2, \qquad  \text{if } \nslot > \nslot^* \\
\end{cases}.
\end{equation}
Thus, the slot access probability also depends on the slot number
\begin{equation}
\paccess =
\begin{cases}
\paccess_1, \qquad  \text{if } \nslot \leq \nslot^* \\
\paccess_2, \qquad  \text{if } \nslot > \nslot^* \\
\end{cases}
\end{equation}
with $\paccess_1= \beta_1/\nuser$ and $\paccess_2= \beta_2/\nuser$.
As a consequence, the slot degree distribution will also vary with the slot number. For ${\nslot \leq \nslot^* }$
we have
\[
\Omega_i= \binom{\nuser}{i} \paccess_1^i (1-\paccess_1)^{\nuser-i}
\]
whereas for $\nslot > \nslot^*$ we have
\[
\Omega_i= \frac{\nslot^*}{\nslot} \binom{\nuser}{i} \paccess_1^i (1-\paccess_1)^{\nuser-i} + \frac{\nslot-\nslot^*}{\nslot} \binom{\nuser}{i} \paccess_2^i (1-\paccess_2)^{\nuser-i}.
\]

The scheme considered in this section cannot be modelled exactly using Theorem~\ref{theorem:rec}, as Theorem~\ref{theorem:rec}  models the case where the slot access probability,  $\paccess$, is chosen independently at random for every slot, while in the considered scheme, $\paccess$ depends on the slot number.
However, the expected performance of the scheme derived through simulations and through Theorem~\ref{theorem:rec} matches down to simulation error for the examples considered in this paper.

In order to illustrate the usefulness of the finite-length analysis derived in Section~\ref{sec:finiteLA}, we shall use the analysis to minimize the \ac{PER}, $\per$. We will consider frameless ALOHA schemes in which in $\beta_1=\betamax$ and $\nslot^*= \nslot_{\max}$. The parameter $\beta_2$ is then chosen as the one that minimizes the packet error rate $\per$ at $\nslot/\nuser =2$.  The optimization was carried out for $\nuser=50, 100$ and $200$. The parameters obtained through the optimization can be found in Table~\ref{tab:params}.
It can be observed that $\beta_2$ grows approximately logarithmically  with $\nslot$. This result is inline with known results from LT code literature, where, in order for LT decoding to be successful with high probability, the average output degree (the equivalent of $\beta$ for LT codes) needs to grow as $\mathcal{O}(\log(k))$, where $k$ is the number of input symbols \cite{luby02:LT}.

Fig.~\ref{fig:per_opt} shows the \ac{PER} for frameless ALOHA schemes optimized for peak throughput (see previous Section) and schemes with two phases, where in the second phase $\beta_2$ is chosen to minimize $\per$. It can be observed how the \ac{PER} becomes substantially lower in this second phase.
For clarity of presentation, Fig.~\ref{fig:per_opt} shows only the \ac{PER} values obtained using the finite-length analysis.
We stress again that Monte Carlo simulations were carried out in order to verify the results and, although the model is not exact in this case, the match is as good as in the rest of the examples in this paper.

\begin{figure}[t]
\centering
	\includegraphics[height=7.5cm]{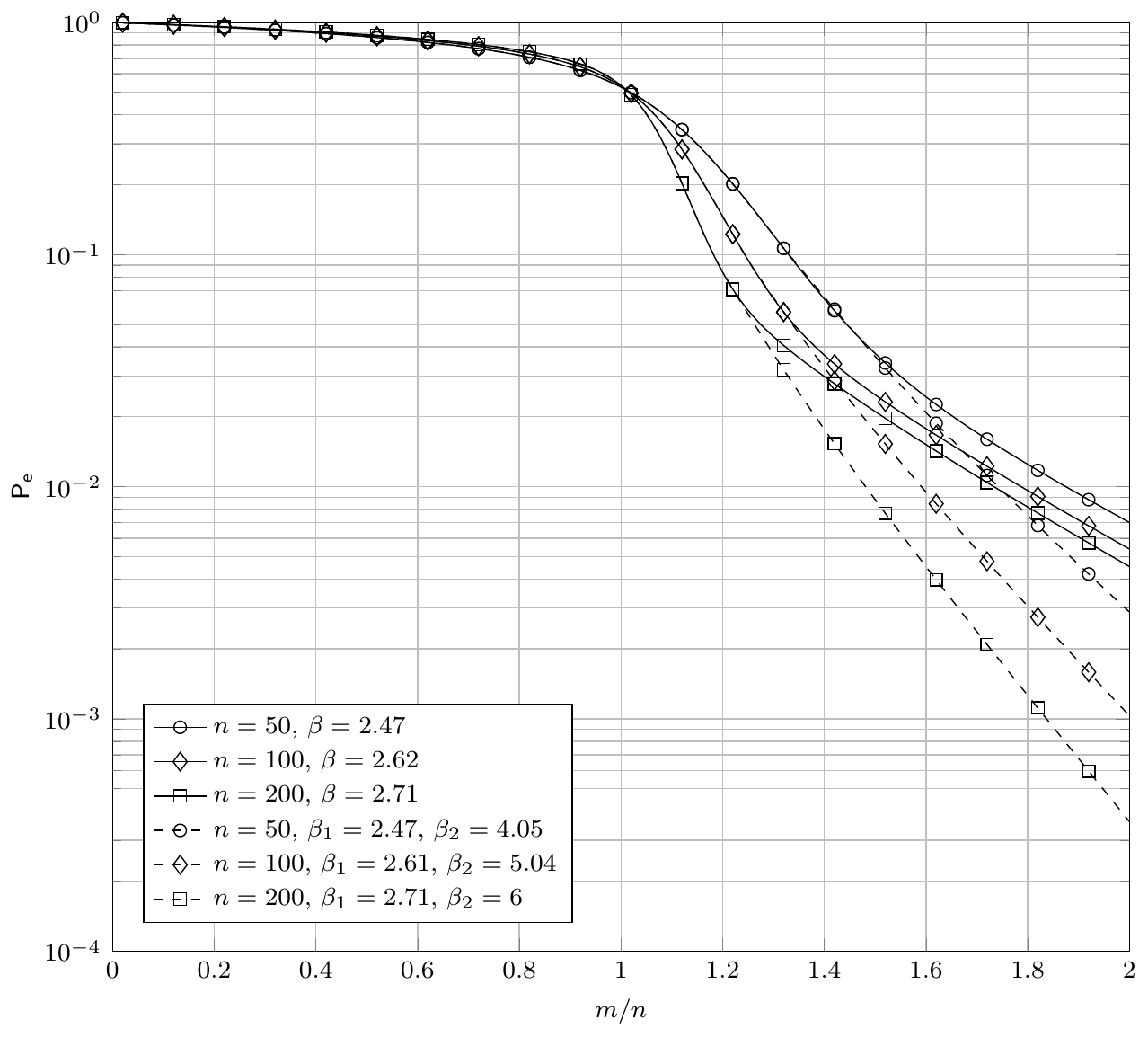}
	\caption{Throughput $\throughput$ as a function of $\nslot/\nuser$ for $\nuser=50, 100$ and $200$. The solid lines represent the result for $\beta=2.5$, the dashed lines represent the result for frameless ALOHA schemes with an initial phase with ${ \beta_1=\betamax }$ and a second phase with $\beta_2$ optimized to minimize $\per$ at ${\nslot = 2 \cdot \nuser}$. }
	\label{fig:per_opt}
\end{figure}

\begin{table}[t]
\centering
\caption{Optimal parameters for two-stage frameless ALOHA}
    \begin{tabular}{|c|c|c|c|}
    \hline
    $\nuser$ & 50 & 100 & 200 \\ \hline \hline
    $\beta_1$ & 2.47 & 2.62 & 2.71 \\ \hline
    $\beta_2$ & 4.05 & 5.04 & 6 \\ \hline
    $\nslot^*$& 66 & 126 & 240 \\ \hline

    \end{tabular}
\label{tab:params}
\end{table}

The scenario chosen for this optimization, trying to maximize the peak throughput in the first phase and then in the second phase pushing down the error floor, might seem somewhat arbitrary. However, it illustrates the fact that the proposed analysis can be used in many ways to tailor frameless ALOHA to specific needs. For example one could define 3 or more phases or even change $\beta$, and hence $\paccess$, slot by slot.

\section{Conclusions}\label{sec:Conclusions}

In this paper we have presented an exact finite-length analysis of frameless ALOHA based on dynamical programming approach. The analysis builds up on results from finite-length analysis of rateless codes. In contrast to most of the works in literature, the analysis applies not only in the error floor region but also in the waterfall region. Simulations have been performed to verify the analysis.  Furthermore, two examples have been presented that illustrate how the proposed analysis can be used to optimize the parameters of frameless ALOHA.

\section*{Acknowledgement}
The research presented in the paper was supported in part by the Danish Council for Independent Research, Grant No. DFF-4005-00281 ``Evolving wireless cellular systems for smart grid communications''.

\bibliographystyle{IEEEtran}
\bibliography{IEEEabrv,references}


\end{document}